\documentclass[12pt]{article}

\usepackage{amsmath,amssymb,amstext,amsthm,url,mathptmx}
\usepackage[usenames]{color}

\usepackage{amsthm}
\theoremstyle{plain}
\newtheorem*{theorem}{Theorem}
\newtheorem*{lemma}{Lemma}
\theoremstyle{remark}
\newtheorem*{question}{Question}

\begin{document}


\title{\#P- and $\oplus$P-completeness of counting roots\\ of a sparse polynomial}

\author{Alexey Milovanov\\
{\tt almas239@gmail.com}}

\maketitle

\begin{abstract}
We improve and simplify the result of the part 4 of ``Counting curves and their projections'' (Joachim von zur Gathen, Marek Karpinski, Igor Shparlinski, \cite{gks}) by showing that counting roots of a sparse polynomial over $\mathbb{F}_{2^n}$ is \#P- and $\oplus$P-complete under deterministic reductions.
\end{abstract}

\section{Result}

Consider the field $\mathbb{F}_{2^n}$. Its elements are presented as polynomials from $\mathbb{F}_2[x]$ modulo some irreducible polynomial of degree $n$. This polynomial can be found in time polynomial in $n$, as well as the matrix that related two representation corresponding to different irreducible polynomials~\cite{len}. Therefore, we do not need to specify a choice of the irreducible polynomial  speaking about polynomial reductions.

Consider the following counting problem (SparcePolynomialRoots): given $n$ and a polynomial from $\mathbb{F}_{2^n}[x]$, find the number of its roots in $\mathbb{F}_{2^n}$. The polynomial is given in a sparse representation, i.e., as a list of coefficients and degrees. The size of input is the total bit size of all this information (each coefficient takes $n$ bits).

\begin{theorem}
\label{main}
\textup{SparsePolynomialRoots}  is \#\textup{P}-complete and $\oplus$\textup{P}-complete.
\end{theorem}

In the paper mentioned above~\cite{gks} the authors provide a randomized polynomial reduction of some \#P-complete problem to the problem of counting points on a curve. We improve this result by (1)~providing a deterministic reduction (proving \#P-completeness and $\oplus$P-completeness with respect to deterministic reductions) and (2)~replacing polynomials of two variables by univariate polynomials (this implies the result for curves by adding a dummy variable).

\section{Proof}
We use  \#3SAT (counting the number of satisfying assignments for a 3-CNF) as a standard \#P-complete problem. Consider some 3-CNF $S$. Each clause in $S$ can be converted into a polynomial equation of the form $l_1\cdot l_2 \cdot l_3=0$, where every $l_i$ is a literal ($x_i$ or $1 + x_i$). All variables are elements of $\mathbb{F}_2$ (i.e., bits). We need to reduce this system of polynomial equations to one polynomial equation over $\mathbb{F}_{2^n}$.

Consider a basis $\omega_1, \ldots, \omega_n$ of  $\mathbb{F}_{2^n}$ over $\mathbb{F}_2$. Then  every $x \in \mathbb{F}_{2^n}$ can be represented as
   $$
x = x_1 \omega_1+ \ldots +x_n \omega_n, 
   $$
where $x_i \in \mathbb{F}_2$. First we transform the clauses (conditions on $x_1,\ldots,x_n$) into (sparse) polynomial conditions on $x$, and then show how the resulting system of polynomial equations can be replaced by one equation. 

Every equation in $S$ has the form $l_1\cdot l_2 \cdot l_3=0$, where $l_i$ are literals, so we need to find polynomials $f_i$ such that $f_i(x)=x_i$. In other terms, all $x$ whose $i$th coordinate $x_i$ is zero should be roots of $f_i$, and $f_i$ should be equal to $1$ on the other half of the field (where $x_i=1$). It is enough for our first step, since a product of three polynomials in  sparse representation is again a polynomial in sparse representation whose size is only polynomially bigger.  The following lemma~\cite[Lemma 3.51]{ln} helps.

\begin{lemma}
Assume that $\alpha_1,\ldots,\alpha_k$ are elements of $\mathbb{F}_{2^n}$ that are linearly independent over $\mathbb{F}_2$. Then the determinant
$$\begin{vmatrix}
\alpha_1 & \alpha_1^2  & \alpha_{1}^4 & \dots & \alpha^{2^{k-1}}\\
\alpha_2 & \alpha_2^2 & \alpha_{2}^{4} & \dots & \omega_{2}^{2^{k-1}} \\ 
\hdotsfor{5} \\
\alpha_n & \alpha_{n}^2 & \alpha_{n}^{4} & \dots & \alpha_{k}^{2^{k-1}}
\end{vmatrix} $$
is a non-zero element of $\mathbb{F}_{2^n}$.
\end{lemma}

\begin{proof}[Proof of the lemma]
Consider this determinant as a function of $\alpha_1$ when other $\alpha_i$ are fixed. In other words, consider the polynomial $P(x)$ that is obtained if we replace $\alpha_1$ by $x$ everywhere in the fist row. We get a polynomial of degree (at most) $2^{k-1}$. The powers of $x$ appearing in $P$ are $1,2,4,\ldots,2^{k-1}$, so this polynomial is linear over $\mathbb{F}_2$ (recall that $(a+b)^2=a^2+b^2$ over a field of characteristic $2$). It has roots $\alpha_2,\ldots,\alpha_k$ (two equal rows guarantee the zero determinant); all $2^{k-1}$ linear combinations of $\alpha_2,\ldots,\alpha_k$ are also roots due to linearity. Reasoning by induction, we may assume that the leading coefficient of $P$, begin the determinant of the same type for smaller $k$, is not zero. Then we know that $P$ has no other roots, and $P(\alpha_1)\ne 0$.
\end{proof}

Now we can define the polynomial
$$ f_1(x) :=c\begin{vmatrix}
x & x^2  & x^4 & \dots & x^{2^{n-1}}\\
\omega_2 & \omega_2^2 & \omega_{2}^{4} & \dots & \omega_{2}^{2^{n-1}} \\ 
\omega_3 & \omega_3^2 & \omega_{3}^{4} & \dots & \omega_{3}^{2^{n-1}} \\ 
\hdotsfor{5} \\
\omega_n & \omega_{n}^2 & \omega_{n}^{4} & \dots & \omega_{n}^{2^{n-1}}
\end{vmatrix} $$
for suitable $c\ne 0$. We know (see the proof of the lemma) that $f_1$ equals $0$ on the linear combinations of $\omega_2,\ldots,\omega_n$, i.e., on all elements with $x_1=0$. Lemma says that $f_1(\omega_1)\ne0$, and the linearity guarantees that $f_1$ has the same values on all elements $x$ with $x_1=1$. It remains to choose $c$ to make $f(\omega_1)$ equal to $1$.

Let us return to our goal: we know now that the number of satisfying assignments for $S$ in $\mathbb{F}_2^n$ is equal to the number of solutions of the system of polynomial equations $P_1(x)=0,P_2(x)=0,\ldots$ in $\mathbb{F}_{2^n}$; each $P_k$ is a product of three polynomials chosen among $f_i$ and $1+f_i$. The number of equations equals the number of clauses. Assume for a while that it is at most $n$ (does not exceed the number of variables). Then we can replace the system by one equation
$$
P_1(x)\omega_1+P_2(x)\omega_2+\ldots = 0
$$
in $\mathbb{F}_{2^n}$ using the fact that polynomials $P_i$ may only have values $0$ and $1$ (being a product of three polynomials with this property).

This finishes the proof of the theorem for the case when the number of variables does not exceed the number of clauses. The general case can be reduced to this special case by adding dummy variables $y_1,\ldots,y_{2s}$ and ``clauses'' $y_1\land y_2=1$, $y_3\land y_4=1$, etc. There are two variables per ``clause'', so this helps. Note also that these ``clauses'' also can be transformed into polynomial equations in the same ways as real clauses (they have conjunction instead of disjunction and $1$ instead of $0$, but this does not matter).

This finishes the proof of our main result.

\section{Remarks and open questions}

We consider sparse polynomials of exponentially large degree. What if we require the degree to be polynomially bounded, in other words, represent the polynomial as an array of coefficients? The question may be asked for polynomials of two variables and corresponding curves.

\begin{question}

Is the problem of finding the number of points on a curve of polynomial-bounded degree \#P-complete? 

Is it $\oplus$P-complete?

Does it belong to polynomial hierarchy?

Is it AM-simple?

\end{question}

May be results of Algebraic Geometry like Fulton's Trace Formula

(\url{http://math.stanford.edu/~dlitt/exposnotes/fultontrace.pdf}) 
 
could help to answer positively the last two questions.

\section*{Acknowlegments}
I would like to thank Alexander Shen for help in writing this paper.

\end{document}